\documentclass[sigconf]{acmart}
\setcopyright{rightsretained}
\usepackage[T1]{fontenc}
\usepackage[utf8]{inputenc}
\usepackage{enumerate}
\newtheorem{mydef}{{Definition}}
\newtheorem{mythm}[mydef]{{Theorem}}
\newtheorem{mylem}[mydef]{{Lemma}}
\newtheorem{myexa}[mydef]{{Example}}

\newtheorem{mypro}[mydef]{{Proposition}}
\begin{document}
\title{Block SOS Decomposition}
\author{Haokun Li}
\affiliation{
  \institution{LMAM \& School of Mathematical Sciences, Peking University}
  \city{Beijing}
  \state{China}
  \postcode{100871}
}
\email{ker@pku.edu.cn}
\author{Bican Xia}
\affiliation{
  \institution{LMAM \& School of Mathematical Sciences, Peking University}
  \city{Beijing}
  \state{China}
  \postcode{100871}
}
\email{xbc@math.pku.edu.cn}
\begin{abstract}
A widely used method for determining whether a multivariate polynomial is a sum of squares of polynomials (SOS), called SOS decomposition, is to decide the feasibility of corresponding semi-definite programming (SDP) problem which can be efficiently solved in theory. In practice, although existing SDP solvers can work out some problems of big scale, the efficiency and reliability of such method decrease greatly while the input size increases. Recently, by exploiting the sparsity of the input SOS decomposition problem, some preprocessing algorithms were proposed \cite{Dai2015,DBLP:conf/cdc/PermenterP15}, which first divide the input problem satisfying special properties into smaller SDP problems and then pass the smaller ones to SDP solvers to obtain reliable results efficiently. A natural question is that to what extent the above mentioned preprocessing algorithms work. That is, how many polynomials satisfying those properties are there in the SOS polynomials? In this paper, we define a concept of {\em block SOS decomposable polynomials} which is a generalization of those special classes of polynomials in \cite{Dai2015} and \cite{DBLP:conf/cdc/PermenterP15}. Roughly speaking, it is a class of polynomials whose SOS decomposition problem can be transformed into smaller ones (in other words, the corresponding SDP matrices can be block-diagnolized) by considering their supports only (coefficients are not considered). Then we prove that the set of block SOS decomposable polynomials has measure zero in the set of SOS polynomials. That means if we only consider supports (not with coefficients) of polynomials, such algorithms decreasing the size of SDPs for those SDP-based SOS solvers can only work on very few polynomials. %As a result, this shows that the SOS decomposition problems that can be optimized by the above mentioned preprocessing algorithms are very few.
\end{abstract}
\maketitle

\section{Introduction}

Given a multivariate polynomial $f$ with real coefficients, consider whether there exist polynomials $q_1,\ldots,q_s$ satisfying $f=\sum_{i=1}^sq_i^2$. This is the so-called SOS decomposition problem which is a very famous and important problem originated from Hilbert's famous work and speech \cite{hilbert1,hilbert2}. It is well-known that Artin %\index{Artin, E.}
proved Hilbert's conjecture \cite{artin}, {\it i.e.}  gave a positive answer to Hilbert's 17th problem whether a positive semi-definite polynomial can be written as sum of squares of rational functions. Artin's theory and method, which is now called the Artin-Schreier theory, lead to the development of modern real algebra and real algebraic geometry. Two major milestones of the subject are quantifier elimination theory over real closed fields established by Tarski \cite{ta51} and the Positivstellensatz discovered by Krivine \cite{krivine1964anneaux} and Stengle \cite{Stengle}.

Although the number of positive semi-definite polynomials is much larger than that of SOS polynomials \cite{blekherman2006there}, replacing nonnegative polynomial constraints with SOS ones and thus requiring an SOS representation of a multivariate polynomial is a key step of many applications, see for example \cite{vand96,lasserre2001global,Parrilo03,kim05,DBLP:journals/siamjo/Schweighofer05}. There have been effective algorithms for determining whether a polynomial can be represented as sum of squares of polynomials, see for example \cite{choi95,powers,parrilo,lasserre2001global,parrilo2003minimizing}. By Gram matrix representation \cite{choi95}, the problem of SOS decomposition can be transformed into an SDP problem, which can be solved symbolically or numerically. Numerical algorithms for SOS decompositions can handle big scale problems and may be used to get exact results \cite{DBLP:conf/issac/KaltofenLYZ08}. %One main numerical method to solve \sos\ decomposition problem is to convert it to {\small SDP} problem.
Actually, there exist some well-known free available SOS solvers which are based on numerical SDP solvers \cite{1393890,DBLP:journals/corr/PapachristodoulouAVPSP13,DBLP:journals/corr/Seiler13}. %\cite{PPP02,lofberg2004yalmip,seiler2013sosopt}.

Although existing SDP solvers can work out some problems of big scale, the efficiency and reliability of such method decrease greatly while the input size increases \cite{DBLP:journals/jsc/DaiGXZ17}. Therefore, decreasing the size of corresponding SDP matrix is an important way to optimize efficiency and reliability of tools based on SDP to solve SOS decomposition problem. Recently, by exploiting the sparsity of the input SOS decomposition problem, some preprocessing algorithms were proposed \cite{Dai2015,DBLP:conf/cdc/PermenterP15}, which first divide the input problem satisfying special properties into smaller SDP problems and then pass the smaller ones to SDP solvers to obtain reliable results efficiently. In \cite{Dai2015}, Dai and Xia defined a class of polynomials called {\em split polynomial}. A split polynomial can be decomposed into several smaller sub-polynomials such that the original polynomial is SOS if and only if the sub-polynomials are all SOS. Thus the original SOS problem can be decomposed equivalently into smaller sub-problems. In \cite{DBLP:conf/cdc/PermenterP15}, Permenter and Parrilo defined the {\em minimal-coordinate-projection} of a given SDP and presented a polynomial time algorithm to find minimal-coordinate-projection. Indeed, finding the minimal-coordinate-projection is the most efficient way by now to take use of the sparsity of SOS problems and to decrease the size of input to SDP solvers for those SDP based SOS solvers.

%For instance,on polynomial $B(x;m):=(\sum_{i=1}^{3m+2}x_i)^2 -2\sum_{i=1}^{3m+2} x_i \sum_{j=0}^m x_{i+3j+1}$, \cite{Baston1969} studied on $B(x;m)$. At this moment, a problem worthy considering is that how many SOS decomposition problem can be optimized by such work.

A natural question is that to what extent the above mentioned preprocessing algorithms work. That is, how many polynomials satisfying those definitions or properties are there in the SOS polynomials? In this paper, we define a concept of {\em block SOS decomposable polynomials} which is a generalization of those special classes of polynomials in \cite{Dai2015} and \cite{DBLP:conf/cdc/PermenterP15}. Roughly speaking, it is a class of polynomials whose SOS decomposition problem can be transformed into smaller ones (in other words, the corresponding SDP matrices can be block-diagnolized) by considering their supports only (coefficients are not considered). Then we prove that the set of block SOS decomposable polynomials has measure zero in the set of SOS polynomials. That means if we only consider supports (not with coefficients) of polynomials, such algorithms decreasing the size of SDPs for those SDP-based SOS solvers can only work on very few polynomials. As a result, this shows that the SOS decomposition problems that can be optimized by the above mentioned preprocessing algorithms are very few.

The rest part of this paper is organized as follows. The concepts of block SOS decomposition, (proper) block SOS support and block SOS decomposable polynomial (Definition \ref{def:square_split-able}) are given in Section \ref{sec:block}. It is also shown in this section that split polynomials are block SOS decomposable polynomials (Proposition \ref{pro:split}) and a polynomial having minimal coordinate projection of corresponding SDP is block SOS decomposable (Proposition \ref{pro:minimal}). The set $\sum_{i=1}^m\sum\mathbb{R}[\bar{x}]_{Q_i}^2$ being closed in $\mathbb{R}[\bar{x}]_{2k}$ (Theorem \ref{thm:Qiclose}) is proven in Section \ref{sec:closed}. %This shows block SOS support can be passed on to sub-polynomial (Corollary \ref{cor:s1s2}). Also we show SOS decomposition of a certain class of polynomials is unique(Lemma \ref{lem:(a-b)2}).
Section \ref{sec:few} proves that the set of block SOS decomposable polynomials in $\mathbb{R}[\bar{x}]_{2k}$ has measure zero in $\sum \mathbb{R}[\bar{x}]_k^2$ (Theorem \ref{thm:zero}).
The paper is concluded in Section \ref{sec:conclusion}.

\section{Block SOS Decomposition}\label{sec:block}

\subsection{Notations}

In this paper, $\bar{x}$ is $n$-dimensional where $n$ is a constant. $\mathbb{R}$ denotes the field of real numbers and $\mathbb{N}$ denotes the ring of integers.
For a mapping $\varphi$, let $\operatorname{rng}(\varphi)$ denote the range of $\varphi$ and $\operatorname{ker}(\varphi)$ denote the kernel of $\varphi$.

For polynomial $f\in \mathbb{R}[\bar{x}]$, the support of $f$, denoted by $S(f)$, is $$\{\alpha\in\mathbb{N}^n~|~ {\rm the~ coefficient~ of}~ \bar{x}^\alpha~ {\rm is~ not~} 0~ {\rm in}~ f\}.$$ $f$ is SOS if and only if there exist polynomials $q_1,\ldots,q_s\in\mathbb{R}[\bar{x}]$ satisfying $f(\bar{x})=\sum_{i=1}^sq_i^2$.

For a set $X\subseteq\mathbb{N}^n$, let $\operatorname{conv}(X)$ denote the convex hull of $X$.
For a set $A\subseteq\mathbb{R}[\bar{x}]$, let $A^2:=\{f^2~|~f\in A\}$ and
$$\sum A^2:=\left\{\sum_{i=1}^sf_i^2~|~f_i \in A~ {\rm for}~ i=1,\ldots,s~ {\rm and}~ s\in\mathbb{N}\right\}.$$ So the set of all polynomials that are SOS can be denoted by $\sum\mathbb{R}[\bar{x}]^2$.

For vector spaces $A,B$, let $$A\oplus B:=\{(a,b)|a\in A \land b\in B\},$$ $$\bigoplus_{i=1}^mA_i:=A_1\oplus\cdots\oplus A_m.$$
Let ${\mathbb S}^{n\times n}$ denote the set of $n\times n$ symmetric matrices and ${\mathbb S}_+^{n\times n}$ denote the set of $n\times n$ positive semi-definite matrices. $\{0,1\}^{n\times n}$ denotes the $n\times n$ matrices with $0,1$ as elements.
For vectors $a$ and $b$, let $\langle a,b\rangle$ denote the inner product of $a$ and $b$. For two $n\times n$ matrices $A, B$, let $\langle A, B\rangle:=\sum_{i,j}A_{ij}B_{ij}=\operatorname{Tr}(A^{\rm T}B).$
For a given $k\in \mathbb{N}$ and a given $Q \subseteq \mathbb{N}^n$, let
\begin{align*}
    \Lambda(k):&=\{(a_1,...,a_n)\in\mathbb{N}^n~|~\sum_{i=1}^{n}a_i\leq k\},\\
    \mathbb{R}[\bar{x}]_k:&=\{f\in\mathbb{R}[\bar{x}]~|~\deg(f)\leq k\},\\
    \mathbb{R}[\bar{x}]_Q:&=\{f\in\mathbb{R}[\bar{x}]~|~S(f)\subseteq Q\},\\
    \sigma(Q):&=\{a\in Q~|~\forall b,c\in Q, b+c\neq 2a\lor b=c=a\}.
\end{align*}
So $\sum\mathbb{R}[\bar{x}]_Q^2$ is the set of polynomials which can be represented as sum of squares of the polynomials in $\mathbb{R}[\bar{x}]_Q.$

\subsection{Main concept}
\begin{mydef}[\cite{Dai2015}]
Suppose $f \in \mathbb{R}[\bar{x}]$ and $Q \subseteq \mathbb{N}^n$. If
$$\text{f is SOS} \Rightarrow \exists \{f_i\}_{i=1}^{s}(f=\sum_{i=1}^sf_i^2 \land S(f_i)\subseteq Q),$$
then $Q$ is said to be an {\em SOS support} of $f$ and denoted by $\operatorname{SOSS}(f,Q)$.
\end{mydef}

The following definition is a general expression of the coefficient-independent optimization of SOS decomposition problems.

\begin{mydef} \label{def:square_split-able}
%Let $k\in\mathbb{N}$ and $T\subseteq\mathbb{N}^n$.
Given $T\subseteq\mathbb{N}^n$ and $Q_i\subseteq\mathbb{N}^n$ for $i=1,\ldots,m$. The set $\{Q_i\}_{i=1}^{m}$ is said to be a {\em block SOS support} of $T$ if
\begin{enumerate}
\item $\{f~|~f\in\sum\mathbb{R}[\bar{x}]^2, S(f)=T\}\subseteq\sum_{i=1}^m\sum\mathbb{R}[\bar{x}]_{Q_i}^2$ and
\item $Q_i\cap Q_j=\emptyset ~{\rm for}~ i\neq j$.
\end{enumerate}

A block SOS support of $S(f)$ where $f\in\mathbb{R}[\bar{x}]$ is also called a {\em block SOS support} of $f$.

Suppose $\{Q_i\}_{i=1}^{m}$ is a block SOS support of $T\subseteq\mathbb{N}^n.$ If for any $Q_i$, $Q_i\neq\emptyset$ and \[\exists f\in\mathbb{R}[\bar{x}] (S(f)=T\land\lnot\operatorname{SOSS}(f,Q_i)),\] then $\{Q_i\}_{i=1}^{m}$ is said to be a {\em proper block SOS support} of $T$ and $T$ is {\em block SOS decomposable}. Similarly, for $f\in\mathbb{R}[\bar{x}]$, if $\{Q_i\}_{i=1}^{m}$ is a proper block SOS support of $S(f)$, then $\{Q_i\}_{i=1}^{m}$ is said to be a {\em proper block SOS support} of $f$ and $f$ is {\em block SOS decomposable}.
\end{mydef}

\begin{myexa}\label{exa:f}
Suppose $T=\{(0),(4),(8)\}$. Then $\{Q_1,Q_2\}$ where $Q_1=\{(0),(4)\}$ and $Q_2=\{(2)\}$ is a proper block SOS support of $T$.
\end{myexa}
\begin{proof}
For any $f\in \mathbb{R}[x]$ with $S(f)=T$, we denote $f=ax^8+bx^4+c$ where $a\neq 0,b\neq 0,c\neq 0$. Then
\[\begin{array}{rl}
f ~ {\rm is~ SOS} & \Leftrightarrow f\geq 0 \quad({\rm since}~ f ~ {\rm is~ a ~ univariate~ polynomial})\\
 & \Rightarrow a\geq 0 \land c\geq 0 \land f(\frac{\sqrt[8]{c}}{\sqrt[8]{a}})\geq 0\\
 & \Rightarrow a\geq 0 \land c\geq 0 \land 2c+b\frac{\sqrt{c}}{\sqrt{a}}\geq 0\\
 & \Rightarrow a\geq 0 \land c\geq 0 \land b\geq -2\sqrt{a}\sqrt{c}.
\end{array}\]
Thus, if $b\geq 0$ and $f$ is SOS, then $f=(\sqrt{a}x^4)^2+(\sqrt{b}x^2)^2+\sqrt{c}^2$.

If $b\leq 0$ and $f$ is SOS, there exist $a' (0\leq a'\leq \sqrt{a})$ and $c' (0\leq c'\leq \sqrt{c})$ such that $b=-2a'c'$ and \[f=(\sqrt{a-{a'}^2}x^4)^2+\sqrt{c-{c'}^2}^2+(a'x^4+c')^2.\]
That proves $\{Q_1,Q_2\}$ is a block SOS support of $T$. Moreover, it is not hard to see that neither $\operatorname{SOSS}(f,Q_1)$ nor $\operatorname{SOSS}(f,Q_2)$ holds. Thus $\{Q_1,Q_2\}$ is a proper block SOS support of $T$.
\end{proof}

\subsection{Related work}
We show in the rest of this section that the optimization methods in \cite{Dai2015} and \cite{DBLP:conf/cdc/PermenterP15} can be restated as finding block SOS support. For reader's convenience, we briefly introduce the related concepts and results of \cite{Dai2015} and \cite{DBLP:conf/cdc/PermenterP15} here. Roughly speaking, the methods in \cite{Dai2015} and \cite{DBLP:conf/cdc/PermenterP15} deal with polynomials whose corresponding SDP matrices can be block-diagonalized while, by Definition \ref{def:square_split-able}, all such polynomials are block SOS decomposable.

\begin{mydef}[Definition 6 of \cite{Dai2015}]
Let $Q,R\subseteq \mathbb{N}^n$. For any $\alpha \in Q+Q$, define:
\begin{align*}
  \varphi_Q(\alpha)&=\{\beta \in Q~|~  \exists \gamma \in Q\: (\beta+ \gamma =\alpha) \},\\
  H_Q(\alpha,R)&=\begin{cases}
    \emptyset  &\alpha \in R;\\
    \{\frac{1}{2}\alpha\} &\alpha \notin R \land \varphi_Q(\alpha)=\{\frac{1}{2}\alpha\};\\
    H_Q'(\alpha,R) & {\rm otherwise};
  \end{cases},\\
  \psi_Q(\alpha)&=H_Q(\alpha,\emptyset),
\end{align*}
${\rm where}~ H_Q'(\alpha,R)=\bigcup_{\beta\in Q ,\exists \gamma \in Q\:( \beta +\gamma =\alpha) }H_Q(2\beta ,R \cup \{\alpha\}).$
\end{mydef}

\begin{mydef}[Definition 8 of \cite{Dai2015}] \label{Dai:D8}
Let $Q$ be a finite set satisfying $\operatorname{SOSS}(p,Q)$ for a polynomial $p$. If there exist some pairwise disjoint nonempty subsets $T_i (i=1,\ldots,u)$ of $\sigma(Q)$ such that
\begin{enumerate}
  \item $\psi_Q(\alpha+\beta)\subseteq T_i$ for any $\alpha, \beta \in \{\gamma |\gamma \in Q ,\psi_Q(2\gamma )\subseteq T_i\}$ holds for any $i=1,\ldots,u$, and
  \item for any $\alpha \in S(p)$, there exists exact one $T_i$ such that $\psi_Q(\alpha) \subseteq T_i$,
\end{enumerate}
then $p$ is said to be a {\em split polynomial} with respect to $T_1,\ldots,T_u$.

If $p$ is a {split polynomial} with respect to a non-empty set $T \subseteq \sigma(Q)$ and its complement in  $\sigma(Q)$, we simply say that $p$ is a split polynomial with respect to $T$.
\end{mydef}

\begin{mythm}[Theorem 4 of \cite{Dai2015}] \label{Dai:T4}
Suppose $p=\Sigma c_{\alpha}x^{\alpha}$ is a split polynomial with respect to $T_1,\ldots,T_u$, then $p$ is SOS if and only if each $p_i=\Sigma_{\alpha \in S(p), \psi_Q(\alpha) \subseteq T_i} c_{\alpha}x^{\alpha}$ is SOS for $i=1,\ldots, u$.
%Especially, for each $p_i$, $\operatorname{SOSS}(p_i,Q_i)$ holds where $Q_i=\{\alpha\in Q|\psi_Q(2\alpha)\subseteq T_i\}$. \emph{(By the proof of Theorem 3 in \cite{Dai2015})}
\end{mythm}

\begin{mypro}\label{pro:split}
Let $Q$ be a finite set satisfying $\operatorname{SOSS}(p,Q)$ for a polynomial $p\in {\mathbb R}[\bar{x}]$. If $p$ is a split polynomial with respect to subsets $T_1,\ldots,T_u$ of $\sigma(Q)$, then $\{Q_i\}_{i=1}^{u}$ is a block SOS support of $p$ where $Q_i=\{\alpha\in Q~|~\psi_Q(2\alpha)\subseteq T_i\}.$
\end{mypro}
\begin{proof}
First, by the proof of Theorem 3 in \cite{Dai2015}, let $p_i$ be as in Theorem \ref{Dai:T4}, $\operatorname{SOSS}(p_i,Q_i)$ holds.
By Definition \ref{Dai:D8} and Theorem \ref{Dai:T4}, we know that if $p$ is a split SOS polynomial, then
\[p\in\sum_{i=1}^u\sum\mathbb{R}[\bar{x}]^2_{Q_i}.\]
And again from the definition and proof there, we know the property is only related to $S(p)$ and irrelevant to the coefficients of $p$.
So \[\{f\in\sum\mathbb{R}[\bar{x}]^2~|~S(f)=S(p)\}\subseteq \sum_{i=1}^u\sum\mathbb{R}[\bar{x}]^2_{Q_i}.\]
%Therefore, $\{\alpha\in S(p)|\psi_Q(\alpha)\subseteq T_i\}_{i=1}^{u}$ is a block SOS support of $S(p)$ and thus a block SOS support of $p$.
The conclusion follows since $\{Q_i\}_{i=1}^{u}$ are pairwise disjoint.
\end{proof}

\begin{mydef}[Definition 1 of \cite{DBLP:conf/cdc/PermenterP15}]\label{PP:D1}
Let $A\subseteq {\mathbb S}^{n\times n}$ and $C\in {\mathbb S}^{n\times n}$. For an SDP
\begin{align*}
  \text{minimize}& \qquad C\cdot X\\
  \text{subject to}&\qquad X\in A\cap {\mathbb S}_+^{n \times n},
\end{align*}
a linear map $P: {\mathbb S}^{n\times n} \rightarrow {\mathbb S}^{n \times n}$ is a {\em coordinate projection} if
\begin{enumerate}[\qquad (a)]
  \item $P({\mathbb S}_+^{n \times n})\subseteq {\mathbb S}_+^{n\times n}$, {\it i.e}, $P$ is a positive map;
  \item $P(A)\subseteq A$;
  \item $P^*(C)=C$;
  \item $\exists M\in {\mathbb S}^{n \times n}\cap \{0,1\}^{n\times n}$\quad$(P(X)=\langle M, X\rangle)$, (we denote $P$ by $P_M$ if $P(X)=\langle M, X\rangle$),
\end{enumerate}
where $P^*$ is the adjoint mapping of $\,P$. % and  $A\circ B:=\operatorname{Tr}(A^TB)$.

We call $P$ a {\em minimal coordinate projection} of an SDP, if $\operatorname{dim}(\operatorname{rng}(P))$ is minimal over coordinate projection.
\end{mydef}

Notice that the method of minimal coordinate projection covers the method of split polynomial according to \cite{DBLP:conf/cdc/PermenterP15}.

\begin{mypro}[Proposition 1 of \cite{DBLP:conf/cdc/PermenterP15}]\label{PP:P1}
Let $A\subseteq {\mathbb S}^{n \times n}$, $C\in {\mathbb S}^{n \times n}$. For an SDP:
\begin{align*}
  \text{minimize}& \qquad C\cdot X\\
  \text{subject to}&\qquad X\in A\cap {\mathbb S}_+^{n \times n},
\end{align*}
if $P$ is a coordinate projection of the SDP, then the optimal value of the SDP equals  the optimal value of
\begin{align*}
  \text{minimize}& \qquad C\cdot X\\
  \text{subject to}&\qquad X\in A\cap {\mathbb S}_+^{n \times n}\cap L,
\end{align*}
where $L=\operatorname{rng}(P)$.
\end{mypro}

\begin{mylem}[Lemma 1 of \cite{DBLP:conf/cdc/PermenterP15}]
Fix $M\in {\mathbb S}^{n\times n}\cap \{0,1\}^{n\times n}$. Then, the mapping $P_M$ is positive, i.e., $\forall X\in {\mathbb S}_+^{n\times n} (\langle M, X\rangle\in  {\mathbb S}_+^{n\times n})$, if and only if $M$ is positive semi-definite.
\end{mylem}

\begin{mylem}[Lemma 2 of \cite{DBLP:conf/cdc/PermenterP15}]
For $M\in {\mathbb S}^{n\times n}\cap \{0,1\}^{n\times n}$, the following statements are equivalent:
\begin{itemize}
  \item The matrix $M$ is positive semi-definite.
  \item There exist $S_1,\ldots,S_p\in\{0,1\}^{n}$, which are orthogonal to each other, and $M=\sum_{k=1}^pS_kS_k^{\rm T}$. %{\color{red}{what is $Q$?}}
\end{itemize}
\end{mylem}

For a polynomial $f$ and a finite set $Q\subseteq \mathbb{N}^n$  satisfying $\operatorname{SOSS}(f,Q)$. Determining whether $f$ is SOS is equal to determining whether the SDP:
\begin{align*}
  \text{minimize}&\qquad 0\cdot X\\
  \text{subject to}&\qquad X\in {\mathbb S}_+^{\#Q \times \#Q},\\
  &\qquad \langle X,(\bar{x}^\alpha)_{\alpha\in Q}(\bar{x}^\alpha)^{\rm T}_{\alpha\in Q}\rangle=f
\end{align*}
is feasible. The SDP above is said to be the SDP corresponding to $f,Q$.

\begin{mypro}\label{pro:minimal}
Given any polynomial $f$ and a finite set $Q\subset \mathbb{N}^n$  satisfying $\operatorname{SOSS}(f,Q)$. If $S_1,\ldots,S_p\in\{0,1\}^{\#Q}$ are orthogonal to each other and the matrix $M:=\sum_{k=1}^pS_kS_k^{\rm T}$
is a  minimal coordinate projection of the SDP corresponding to $f,Q$. Then $\{Q_k\}_{k=1}^p$ is a block SOS support of $f$ where $Q_k=S(\left<S_k,(\bar{x}^\alpha)_{\alpha\in Q}\right>)$.
\end{mypro}
\begin{proof}
By Proposition \ref{PP:P1}, we know that the SDP is feasible if and only if it is feasible on the minimal coordinate projection. So if $f$ is SOS, then $f\in\sum_{k=1}^p\sum\mathbb{R}[\bar{x}]^2_{Q_k}$. By Theorem 2 of \cite{DBLP:conf/cdc/PermenterP15}, we know that the minimal coordinate projection is irrelevant to the coefficients of $f$. Therefore $$\{f'\in\sum\mathbb{R}[\bar{x}]^2~|~S(f')=S(f)\}\subseteq\sum_{k=1}^p\sum\mathbb{R}[\bar{x}]^2_{Q_k}.$$

%By Definition \ref{def:square_split-able}, $\{Q_k\}_{k=1}^p$ is a block SOS support of $f$.
Because $S_1,\ldots,S_p\in\{0,1\}^{\#Q}$ are orthogonal, $\{Q_k\}_{k=1}^p$ are pairwise disjoint. That completes the proof.
\end{proof}

%The following example is about minimal coordinate projection and block SOS decomposition.
We show in the following example a polynomial whose block SOS support can be computed by finding a minimal coordinate projection of the corresponding SDP but the block SOS support computed by the algorithm of \cite{DBLP:conf/cdc/PermenterP15} is bigger than that obtained in Example \ref{exa:f}.

\begin{myexa}\label{ex:f1}
Let $f=x^8-2x^4+1$. So $S(f)=\{(0),(4),(8)\}$ and, for $Q=\operatorname{conv}(\frac{1}{2}(S(f)\cap \mathbb{N}))=\{(0),(1),(2),(3),(4)\}$, we have $\operatorname{SOSS}(f,Q)$.
By the algorithm in \cite{DBLP:conf/cdc/PermenterP15},
\begin{gather*}
\begin{bmatrix}
  1&0&0&0&1\\
  0&1&0&1&0\\
  0&0&1&0&0\\
  0&1&0&1&0\\
  1&0&0&0&1
\end{bmatrix}
\end{gather*}
is a minimal coordinate projection of the SDP corresponding to $f,Q$.
By the proof of Proposition \ref{pro:minimal}, the corresponding proper block SOS support is $\{\{(0),(4)\},\{(1),(3)\},\{(2)\}\}$.

By Example \ref{exa:f}, we know that $\{\{(0),(4)\},\{(2)\}\}$ is also a proper block SOS support of $f$, which is smaller.
\end{myexa}

\section{$\sum_{i=1}^m\sum\mathbb{R}[\bar{x}]_{Q_i}^2$ is Closed} \label{sec:closed}

For arbitrarily given $k\in \mathbb{N}$ and a finite set $Q \subset \mathbb{N}^n$, when we talk about the topology of $\mathbb{R}[\bar{x}]_k$ in this paper, we mean the topology of finite-dimensional vector space $\mathbb{R}^{\#\Lambda(k)}.$ Similarly, the topology of $\mathbb{R}[\bar{x}]_Q$ is the topology of finite-dimensional vector space $\mathbb{R}^{\#Q}$.

For a given $k\in \mathbb{N}$, we prove in this section that the polynomial space $\sum_{i=1}^m\sum\mathbb{R}[\bar{x}]_{Q_i}^2$ is closed in $\mathbb{R}[\bar{x}]_{2k}$,  where $Q_i\subseteq\varLambda(k)$.

%Before we move on, we are going to prove some lemmas.
\begin{mylem}
  If $A\in {\mathbb S}_+^{m\times m}$, then there exist $v_i\in \mathbb{R}^{m\times 1}\:(i=1,\ldots,m)$ such that $A=\sum_{i=1}^mv_iv_i^{\rm T}$\label{lem:Sfj}.
\end{mylem}
\begin{proof}
By using Cholesky decomposition, we have
$$A=LL^{\rm T}=\sum_{i=1}^m L_iL_i^{\rm T}$$
where $L\in \mathbb{R}^{m\times m}$ and $L_i$ is the $i^{{\rm th}}$ column of $L$. Let $v_i:=L_i$.
\end{proof}

\begin{mylem}\label{lem:xG}
Suppose $Q\subset\mathbb{N}^n$ and $m=\#Q<+\infty$.  Then
$$\sum\mathbb{R}[\bar{x}]_Q^2=\{\langle(\bar{x}^{a+b})_{(a,b)\in Q\times Q}, G\rangle~|~G\in {\mathbb S}_+^{m\times m}\}.$$
\end{mylem}
\begin{proof}
For $G\in {\mathbb S}_+^{m \times m},$ by Lemma \ref{lem:Sfj}, there exist $\{v_i\}_{i=1}^{m}$ such that $G=\sum_{i=1}^mv_iv_i^{\rm T}$.
So
%$\langle(\bar{x}^{a+b})_{(a,b)\in Q\times Q}, G\rangle = \sum_{i=1}^m\langle(\bar{x}^{a+b})_{(a,b)\in Q\times Q}, v_iv_i^{\rm T}\rangle$
%$=\sum_{i=1}^m\langle(\bar{x}^a)_{a\in Q}, v_i\rangle^2\in \sum\mathbb{R}[\bar{x}]_Q^2.$
\begin{align*}
 &\langle(\bar{x}^{a+b})_{(a,b)\in Q\times Q}, G\rangle\\
=&\sum_{i=1}^m\langle(\bar{x}^{a+b})_{(a,b)\in Q\times Q}, v_iv_i^{\rm T}\rangle\\
=&\sum_{i=1}^m\langle(\bar{x}^a)_{a\in Q}, v_i\rangle^2\in \sum\mathbb{R}[\bar{x}]_Q^2.
\end{align*}

On the other hand, for any $\sum_{i=1}^kf_i^2\in\sum\mathbb{R}[\bar{x}]_Q^2,$ there exist $v_i\in \mathbb{R}^{m\times 1}$ such that ${f_i=\langle(\bar{x}^a)_{a\in Q},v_i\rangle}$ for each $i=1,\ldots,k$. Then %It is easy to see that
\begin{align*}
\sum_{i=1}^kf_i^2=&\sum_{i=1}^k\langle(\bar{x}^a)_{a\in Q}, v_i\rangle^2
=\sum_{i=1}^k\langle(\bar{x}^{a+b})_{(a,b)\in Q\times Q},v_iv_i^{\rm T}\rangle\\
=&\langle(\bar{x}^{a+b})_{(a,b)\in Q\times Q},\sum_{i=1}^kv_iv_i^{\rm T}\rangle
=\langle(\bar{x}^{a+b})_{(a,b)\in Q\times Q},LL^{\rm T}\rangle\\
\in&\{\langle(\bar{x}^{a+b})_{(a,b)\in Q\times Q},G\rangle~|~G\in {\mathbb S}_+^{m \times m}\},
\end{align*}
where $L=(v_1,\ldots,v_k)\in \mathbb{R}^{m\times k}.$
%That proves $\{\langle(\bar{x}^{a+b})_{(a,b)\in Q\times Q},G\rangle|G\in {\mathbb S}_+^{m \times m}\} = \sum\mathbb{R}[\bar{x}]_Q^2$.
\end{proof}

\begin{mylem}
  \label{lem:SOSfj}
  Suppose $Q\subset\mathbb{N}^n$ and $m=\#Q<+\infty$.   Then
    $$\sum\mathbb{R}[\bar{x}]_Q^2=\{f_1^2+\cdots+f_{m}^2|f_1,...,f_{m}\in\mathbb{R}[\bar{x}]_Q\}.$$
\end{mylem}
\begin{proof}
  By Lemma \ref{lem:xG}, we know $$\sum\mathbb{R}[\bar{x}]_Q^2=\{\langle(\bar{x}^{a+b})_{(a,b)\in Q\times Q},G\rangle~|~G\in {\mathbb S}_+^{m\times m}\}.$$
  By Lemma \ref{lem:Sfj}, there exist $\{v_i\}_{i=1}^{m}$ such that $G=\sum_{i=1}^{m}v_iv_i^{\rm T}.$

 Let $f_i:=\left<(\bar{x}^a)_{a\in Q},v_i\right>$ for $i=1,\ldots,m$, we have $f=\sum_{i=1}^{m}f_i^2$.
\end{proof}

The proof of the following theorem uses the idea of the proof about $M_k$ is closed in \cite{DBLP:journals/siamjo/Schweighofer05}.
\begin{mythm}
  \label{thm:Qiclose}
  For a given $k\in \mathbb{N}$ and any finite set $\{Q_i\}_{i=1}^m$ where $Q_i\subseteq\varLambda(k)$ for each $i=1,\ldots,m$, $$\sum_{i=1}^m\sum\mathbb{R}[\bar{x}]_{Q_i}^2$$ is closed in $\mathbb{R}[\bar{x}]_{2k}$. % where $Q_i\subseteq\varLambda(k)$ for each $i=1,\ldots,m$.
%Suppose $\{Q_i\}_{i=1,\ldots,m},Q_i\subseteq \Lambda(k),m\in \mathbb{N}$ satisfy $Q_i\cap Q_j=\emptyset,i\neq j$, then $\sum_{i=1}^m\sum\mathbb{R}[\bar{x}]_{Q_i}^2$ is closed in $\mathbb{R}[\bar{x}]_{2k}$.
\end{mythm}
\begin{proof}
  Let $\varphi$ be a continuous mapping:
  \begin{align*}
    \varphi:&\bigoplus_{\#Q_1}\mathbb{R}[\bar{x}]_{Q_1}\times\cdots\times\bigoplus_{\#Q_m}\mathbb{R}[\bar{x}]_{Q_m}\rightarrow\mathbb{R}[\bar{x}]_{2k}\\
    &((f_{1,j})_{j\in Q_1},\ldots,(f_{m,j})_{j\in Q_m})\mapsto\sum_{i=1}^m\sum_{j\in Q_i}f_{i,j}^2.
  \end{align*}

  By Lemma \ref{lem:SOSfj}, we know $\sum_{i=1}^m\sum\mathbb{R}[\bar{x}]_{Q_i}^2=\operatorname{rng}(\varphi)$ and $\operatorname{ker}(\varphi)={0}$.

  Let $V$ be the image of unit sphere of $$\bigoplus_{\#Q_1}\mathbb{R}[\bar{x}]_{Q_1}\times\cdots\times\bigoplus_{\#Q_m}\mathbb{R}[\bar{x}]_{Q_m}.$$ Obviously $V$ is compact and $0\notin V$. Since $\varphi(\lambda z)=\lambda^2\varphi(z)$ for $\lambda\in\mathbb{R}$, we have $\sum_{i=1}^m\sum\mathbb{R}[\bar{x}]_{Q_i}^2=\sum_{\lambda\in [0,\infty)}\lambda V$.

    Given an infinite sequence $\{\lambda_iv_i\}_{i=1,2,\ldots}$ where $\lambda_i\in[0,\infty)$ and $v_i\in V$, assume $\lim_{i\rightarrow\infty}\lambda_iv_i=p$. Because $V$ is compact, there exists a convergent subsequence $\{v_{i_j}\}\subseteq\{v_i\}$ such that $v=\lim_{j\rightarrow\infty}v_{i_j}\in V$. So $v\neq0$.

      Then
        $$\lim_{j\rightarrow\infty}\lambda_{i_j}=\lim_{j\rightarrow\infty}\frac{|\lambda_{i_j}v_{i_j}|}{|v_{i_j}|}=\frac{|p|}{|v|}\geq 0,$$
      where $|\cdot|$ stands for $\|\cdot\|_2$. Finally,
      \begin{align*}
       p&=\lim_{j\rightarrow\infty}\lambda_{i_j}v_{i_j}=(\lim_{j\rightarrow\infty}\lambda_{i_j})(\lim_{j\rightarrow\infty}v_{i_j})=\frac{|p|}{|v|}v\\
        &\in\sum_{\lambda\in [0,\infty)}\lambda V=\sum_{i=1}^m\sum\mathbb{R}[\bar{x}]_{Q_i}^2.
       \end{align*}
      So $\sum_{i=1}^m\sum\mathbb{R}[\bar{x}]_{Q_i}^2$ is closed.
\end{proof}

\section{Block SOS Decomposable Polynomial Is Few}\label{sec:few}

Firstly, we show that  SOS decomposition of a certain class
of polynomials is unique by the two lemmas below.

\begin{mylem}
  \label{lem:sigmaQ}
  For any finite set $Q\subset\mathbb{N}^n$, $$\operatorname{conv}(Q)=\operatorname{conv}(\sigma(Q)).$$
\end{mylem}
\begin{proof}
  Without loss of generality, suppose $Q\neq\sigma(Q)$.
 Let $Q\backslash\sigma(Q)=\{a_1,a_2,\ldots,a_r\}$ and
 $$Q_0=Q, ~ Q_i=Q_{i-1}\backslash\{a_i\}~ {\rm for}~ i=1,\ldots,r.$$
 So $Q_r=\sigma(Q)$. We then prove the claim $$\operatorname{conv}(Q_i)=\operatorname{conv}(Q)~ {\rm for}~ i=0,...,r$$ by induction on $i$.

  When $i=0$, it is obviously true. Suppose the claim holds for $i-1$. It is clear that we need only to prove $a_i\in \operatorname{conv}(Q_i).$

  Since $a_i\in Q\backslash\sigma(Q)$, there exist $b,c\in Q$ such that $$b\ne a_i,~ c\neq a_i, ~ b+c=2a_i.$$
  Since $b,c\in Q\subseteq \operatorname{conv}(Q) = \operatorname{conv}(Q_{i-1})$, there exist $\lambda_{b,\alpha},\lambda_{c,\alpha}\in \left[0,1\right]$ for each $\alpha\in Q_{i-1}$ such that
  \[\begin{array}{ll}
  b=\sum_{\alpha\in Q_{i-1}}\lambda_{b,\alpha}\alpha, & \sum_{\alpha\in Q_{i-1}}\lambda_{b,\alpha}=1,\\
  c=\sum_{\alpha\in Q_{i-1}}\lambda_{c,\alpha}\alpha, & \sum_{\alpha\in Q_{i-1}}\lambda_{c,\alpha}=1.
  \end{array}
  \]
  Because $b\ne a_i$ and $c\neq a_i$, we have that $\lambda_{b,a_i}+\lambda_{c,a_i}<2$.

  On the other hand, $2a_i=\sum_{\alpha\in Q_{i-1}}\lambda_{b,\alpha}\alpha+\sum_{\alpha\in Q_{i-1}}\lambda_{c,\alpha}\alpha$ because $b+c=2a_i$. Therefore
  \[\begin{array}{rl}
  a_i & =\sum_{\alpha\in Q_{i-1}\backslash\{a_i\}}\frac{\lambda_{b,\alpha}+\lambda_{c,\alpha}}{2-\lambda_{b,a_i}-\lambda_{c,a_i}}\alpha\\
    & =\sum_{\alpha\in Q_i}\frac{\lambda_{b,\alpha}+\lambda_{c,\alpha}}{2-\lambda_{b,a_i}-\lambda_{c,a_i}}\alpha.
    \end{array}
  \]
  It is easy to see that $$\frac{\lambda_{b,\alpha}+\lambda_{c,\alpha}}{2-\lambda_{b,a_i}-\lambda_{c,a_i}}\ge 0~ {\rm and}~ \sum_{\alpha\in Q_i}\frac{\lambda_{b,\alpha}+\lambda_{c,\alpha}}{2-\lambda_{b,a_i}-\lambda_{c,a_i}}=1.$$
  This means $a_i\in \operatorname{conv}(Q_i)$ that completes the proof.
  %By $Q_{i-1}\in \operatorname{conv}(Q_i)$, we know $\operatorname{conv}(Q_i)=\operatorname{conv}(Q_{i-1})=\operatorname{conv}(Q)$.
%  So In the induction step, we have shown that $\operatorname{conv}(\sigma(Q))=\operatorname{conv}(Q_r)=\operatorname{conv}(Q)$.
\end{proof}

\begin{mylem}\label{lem:(a-b)2}
  Let $a,b\in\mathbb{N}^n$ and $a\neq b$. If $a-b\in\{-1,0,1\}^n$, then any SOS decomposition of $\bar{x}^{2a}+2\bar{x}^{a+b}+\bar{x}^{2b}$ can only be of the form $\sum_{i=1}^m(\lambda_{i1}\bar{x}^a+\lambda_{i2}\bar{x}^b)^2$ where there exists $i$ such that $\lambda_{i1}\lambda_{i2}\neq0.$
\end{mylem}

\begin{proof}
  Suppose $\{f_i\}_{i=1}^{m}$ satisfy $(\bar{x}^a+\bar{x}^b)^2=\sum_{i=1}^mf_i^2.$ Let $Q=\cup_{i=1}^mS(f_i)$. By the definition of $\sigma(Q)$, it is not hard to see that the squares of items in $\sigma(Q)$ must appear in $\sum_{i=1}^mf_i^2$ (since they do not eliminate). Therefore $\sigma(Q)\subseteq \{a,b\}$. By Lemma \ref{lem:sigmaQ}, we know
  $$Q\subseteq \operatorname{conv}(Q)=\operatorname{conv}(\sigma(Q))\subseteq \operatorname{conv}(\{a,b\})=\{a,b\}.$$
  The last equality holds because $a-b\in\{-1,0,1\}^n$. So $f_i=\lambda_{i1}\bar{x}^a+\lambda_{i2}\bar{x}^b.$ Because $\bar{x}^{a+b}$ appears in $\sum_{i=1}^mf_i^2$, there exists $i$ such that $\lambda_{i1}\lambda_{i2}\neq0.$
\end{proof}

\begin{mylem}\label{pro:s1s2}
  Let $k\in\mathbb{N} ,\, {S_1\subset S_2\subseteq\Lambda(2k)}$. If $$\{f\in\sum\mathbb{R}[\bar{x}]_{k}^2~|~S(f)=S_2\}\neq\emptyset$$ and $\{Q_i\}_{i=1}^{m}$ is a block SOS support of $S_2$, then $\{Q_i\}_{i=1}^{m}$ is also a block SOS support of $S_1$.
\end{mylem}
\begin{proof}
  Take $f_2\in\{f\in\sum\mathbb{R}[\bar{x}]_{k}^2~|~S(f)=S_2\}$. For any $f_1\in\{f\in\sum\mathbb{R}[\bar{x}]_{k}^2~|~S(f)=S_1\}$, without loss of generality, suppose the absolute values of the coefficients of $f_2$ are less than the absolute values of nonzero coefficients of $f_1$. Let $g_j=\frac{f_2}{j}+f_1$ for any positive integer $j$. Then $S(g_j)=S_2$ and $\lim_{j\rightarrow\infty}g_j=f_1$.

   Because $\{Q_i\}_{i=1}^{m}$ is a block SOS support of $S_2$, then $$\{f\in\sum\mathbb{R}[\bar{x}]_{k}^2~|~S(f)=S_2\}\subseteq\sum_{i=1}^m\sum\mathbb{R}[\bar{x}]_{Q_i}^2.$$ By Theorem \ref{thm:Qiclose}, $f_1\in\sum_{i=1}^m\sum\mathbb{R}[\bar{x}]_{Q_i}^2$. So $$\{f\in\sum\mathbb{R}[\bar{x}]_{k}^2~|~S(f)=S_1\}\subseteq\sum_{i=1}^m\sum\mathbb{R}[\bar{x}]_{Q_i}^2.$$

  Therefore, $\{Q_i\}_{i=1}^{m}$ is a block SOS support of $S_1$.
\end{proof}

Based on the above two lemmas, we then prove that, for an arbitrarily given degree bound $2k$, the {\em full polynomials}, {\it i.e.} polynomials having all monomials with degree no more than $2k$, are not block SOS decomposable.

\begin{mythm}\label{pro:lambda2k}
  For any $k\in\mathbb{N}$, $\Lambda(2k)$ is not block SOS decomposable.
\end{mythm}
\begin{proof}
  Suppose $\{Q_i\}_{i=1}^{m}$ is a proper block SOS support of $\Lambda(2k)$. Then for any $Q_i (i=1,\ldots,m)$, we have $\Lambda(k)\not\subseteq Q_i.$

  Denote by $\{e_i\}_{i=1}^{n}$ the unit vectors of $\mathbb{N}^n$. Firstly, we prove there exist $a, b\in\Lambda(k)$ satisfying $a\neq b$, $a-b\in\{-1,0,1\}^n$ and $a\notin Q_i\lor b\notin Q_i$ for any $i=1,\ldots,m.$

  If every $Q_i (i=1,...,m)$ is empty, simply let $a=0, b=e_1$.

  Suppose there exists a non-empty $Q'\in\{Q_i\}_{i=1}^{m}$. Since $\Lambda(k)\not\subseteq Q'$, there exist $\alpha, \beta\in\Lambda(k)$ such that $\alpha\in Q'$ and $\beta\notin Q'$. Let $\{r_j\}_{j=1}^{u}\subset \Lambda(k)$ be a path connecting $\alpha$ ($=r_1$) and $\beta$ ($=r_u$) in $\Lambda(k)$ and satisfying for any $j\in\{1,...,u\}$ there exists $i\in\{1,...,n\}$ such that $r_{j+1}-r_j=\pm e_i$. Then there exists $j_0\in\{1,\ldots,u-1\}$ such that $r_{j_0}\in Q'$ and $r_{j_0+1}\notin Q'$. We set $a=r_{j_0}, b=r_{j_0+1}$. Because $Q_i (i=1,\ldots,m)$ are pairwise disjoint, $a\notin Q_i\lor b\notin Q_i$ for any $i=1,\ldots,m.$

   %So there exist $a,b\in\Lambda(k)$ satisfying $a\neq b$,$a-b\in\{-1,0,1\}^n$ and $\forall i\in{1,..,m}(a\notin Q_i\lor b\notin Q_i)$.
   By Lemma \ref{pro:s1s2}, $\{Q_i\}_{i=1}^{m}$ is also a block SOS support of $\{2a,a+b,2b\}$. So $(\bar{x}^a+\bar{x}^b)^2\in\sum_{i=1}^m\sum\mathbb{R}[\bar{x}]_{Q_i}^2$. But by Lemma \ref{lem:(a-b)2} we see that $(\bar{x}^a+\bar{x}^b)^2=\sum_{i=1}^m(\lambda_{i1}\bar{x}^a+\lambda_{i2}\bar{x}^b)^2$ and there exists $i$ such that $\lambda_{i1}\lambda_{i2}\neq0$. That contradicts with the claim that for any $i\in\{1,\ldots,m\}$, $a\notin Q_i\lor b\notin Q_i$.

    So $\Lambda(k)$ is the only block SOS support of $\Lambda(2k)$.
\end{proof}

\begin{mythm}\label{thm:zero}
For any $k\in\mathbb{N}$, the set of block SOS decomposable polynomials in $\mathbb{R}[\bar{x}]_{2k}$ has measure zero in $\sum\mathbb{R}[\bar{x}]_{k}^2$. %$\mathbb{R}[\bar{x}]_{2k}$.
\end{mythm}
\begin{proof}
By Theorem \ref{pro:lambda2k}, we see that the set of block SOS decomposable polynomials in $\mathbb{R}[\bar{x}]_{2k}$ is a subset of $$\{f\in\mathbb{R}[\bar{x}]_{2k}~|~S(f)\neq\Lambda(2k)\}.$$
So it has measure zero in $\mathbb{R}[\bar{x}]_{2k}$. Meanwhile we know $\sum \mathbb{R}[\bar{x}]_{k}^2$ has measure non-zero in $\mathbb{R}[\bar{x}]_{2k}$ because the cone $\sum \mathbb{R}[\bar{x}]_{k}^2$ is a full-dimensional convex cone in $\mathbb{R}[\bar{x}]_{2k}$, {\it i.e.}, $\sum \mathbb{R}[\bar{x}]_{k}^2-\sum \mathbb{R}[\bar{x}]_{k}^2=\mathbb{R}[\bar{x}]_{2k}$. Thus the set of block SOS decomposable polynomials in $\sum\mathbb{R}[\bar{x}]^2_{k}$ must have measure zero in $\sum\mathbb{R}[\bar{x}]_{k}^2$.
\end{proof}

By the definition of block SOS decomposable polynomial, Theorems \ref{pro:lambda2k} and \ref{thm:zero} indicate that those polynomials whose SDP matrices (corresponding to their SOS decompositions) can be block-diagonalized are very few in the set of SOS polynomials.

\section{Conclusion}\label{sec:conclusion}
Finding split polynomials or minimal coordinate projection \cite{Dai2015,DBLP:conf/cdc/PermenterP15} are two recently proposed methods which, if success, can decrease the size of SDP matrices via block-diagnalization and thus increase greatly the efficiency and reliability of the SDP based SOS decomposition methods. In this paper, we investigate the scope that the two methods work. We first define a class of polynomials, namely, {\em block SOS decomposable polynomials}, which is a generalization of those classes of polynomials in \cite{Dai2015} and \cite{DBLP:conf/cdc/PermenterP15}. Roughly speaking, it is a class of polynomials whose SOS decomposition problem can be transformed into smaller ones by considering their supports only (coefficients are not considered). Then we prove that the set of block SOS decomposable polynomials has measure zero in the set of SOS polynomials. That means if we only consider supports (not with coefficients) of polynomials, such algorithms decreasing the size of corresponding SDP matrices can only work on very few polynomials. As a result, this shows that the SOS decomposition problems that can be optimized by the above mentioned methods are very few.

%On the other hand, we do not have an algorithm for finding block SOS decomposable polynomials and we do not know whether block SOS decomposable polynomials are just those with minimal coordinate projection, neither. As for the efficiency of such kind of methods, the algorithm in \cite{DBLP:conf/cdc/PermenterP15} is the best.
On the other hand, it is interesting to give an algorithm for finding block SOS decomposable polynomials in the future. By now the most efficient way to block-diagnolize corresponding SDP matrices is to find minimal coordinate projection by the algorithm of \cite{DBLP:conf/cdc/PermenterP15}. We do not know whether block SOS decomposable polynomials are just those whose corresponding SDPs have minimal coordinate projections. At least, there are examples (see Examples \ref{exa:f} and \ref{ex:f1}) whose block SOS supports computed via minimal coordinate projection are not minimal.

\section*{Acknowledgement}
This work was supported partly by NSFC under grants 61732001 and 61532019 and CDZ project CAP (GZ 1023).

\bibliographystyle{ACM-Reference-Format}
\bibliography{bsd}

\end{document}